\documentclass{elsarticle}
\usepackage[T1]{fontenc}
\usepackage[utf8]{inputenc}
\usepackage{graphicx}
\usepackage{epsfig}
\usepackage{color}
\usepackage{amssymb,amsmath,amsthm}
\usepackage{tikz}
\usepackage{url}
\usepackage[colorlinks=true,citecolor=black,linkcolor=black,urlcolor=blue]{hyperref}

\def\acc#1{\left\{ #1 \right\}}

\renewcommand{\le}{\leqslant}
\renewcommand{\ge}{\geqslant}

\newtheorem{theorem}{Theorem}

\newtheorem{corollary}[theorem]{Corollary}
\newtheorem{lemma}[theorem]{Lemma}

\begin{document}
\begin{frontmatter}
\title{Vertex partitions of $(C_3,C_4,C_6)$-free planar graphs}
\author[LIRMM]{Fran\c{c}ois Dross}
\author[LIRMM]{Pascal Ochem}
\address[LIRMM]{LIRMM, Universit\'e de Montpellier, and CNRS. France}

\begin{abstract}
A graph is $(k_1,k_2)$-colorable if its vertex set can be partitioned into a graph with maximum degree at most $k_1$ and and a graph with maximum degree at most $k_2$.
We show that every $(C_3,C_4,C_6)$-free planar graph is $(0,6)$-colorable.
We also show that deciding whether a $(C_3,C_4,C_6)$-free planar graph is $(0,3)$-colorable is NP-complete.
\end{abstract}

\end{frontmatter}

\section{Introduction}\label{sec:intro}
A graph is $(k_1,k_2)$-colorable if its vertex set can be partitioned into a graph with maximum degree at most $k_1$ and and a graph with maximum degree at most $k_2$.
Choi, Liu, and Oum~\cite{CLO} have established that there exists exactly two minimal sets of forbidden cycle length such that every planar graph
is $(0,k)$-colorable for some absolute constant $k$. 
\begin{itemize}
 \item planar graphs without odd cycles are bipartite, that is, $(0,0)$-colorable.
 \item planar graphs without cycles of length $3$, $4$, and $6$ are $(0,45)$-colorable.
\end{itemize}

The aim of this paper is to improve this last result.
Notice that forbidding cycles of length $3$, $4$, and $6$ as subgraphs or as induced subgraphs result in the same graph class.
For every $n\ge3$, we denote by $C_n$ the cycle on $n$ vertices.
So we are interested in the class $\cal C$ of $(C_3,C_4,C_6)$-free planar graph.

We will prove the following two theorems in the next two sections.

\begin{theorem} \label{t+}
Every graph in $\cal C$ is $(0,6)$-colorable.
\end{theorem}

\begin{theorem} \label{t-}
For every $k\ge1$, either every graph in $\cal C$ is $(0,k)$-colorable, or deciding whether a graph in $\cal C$ is $(0,k)$-colorable is NP-complete.
\end{theorem}

In addition, we construct a graph in $\cal C$ that is not $(0,3)$-colorable in Section~\ref{sec:ce}. This graph and Theorem~\ref{t-} imply the following.
\begin{corollary}
Deciding whether a graph in $\cal C$ is $(0,3)$-colorable is NP-complete.
\end{corollary}

Since we deal with $(0,k)$-colorings for some $k\ge2$, we denote by the letter $0$ the color of the vertices that induce
the independent set and we denote by the letter $k$ the color of the vertices that induce the graph with maximum degree $k$.

\section{Proof of Theorem~\ref{t+}} \label{sec:t+}
The proof will be using the discharging method.
For every plane graph $G$, we denote by $V(G)$ the set of vertices of $G$, by $E(G)$ the set of edges of $G$, and by $F(G)$ the set of faces of $G$.

Let us define the partial order $\preceq$. Let $n_3(G)$ be the number of $3^+$-vertices in $G$.
For any two graphs $G_1$ and $G_2$, we have $G_1\prec G_2$ if and only if at least one of the following conditions holds:
\begin{itemize}
\item $|V(G_1)|<|V(G_2)|$ and $n_3(G_1)\le n_3(G_2)$.
\item $n_3(G_1)<n_3(G_2)$.
\end{itemize} 
Note that the partial order $\preceq$ is well-defined and is a partial linear extension of the subgraph poset.

We suppose for contradiction that $G$ is a graph in $\cal C$ that is not $(0,6)$-colorable and is minimal according to $\preceq$.
Let $n$ denote the number of vertices, $m$ the number of edges and $f$ the number of faces of $G$.
For every vertex $v$, the degree of $v$ in $G$ is denoted by $d(G)$.
For every face $\alpha$, the \emph{degree} of $\alpha$, denoted $d(\alpha)$,
is the number of edges that are shared between this face and another face, plus twice the number of edges that are entirely in $\alpha$.
More generally, when counting the number of edges of a certain type in a face, we will always count twice the edges that are only in this face.
For all $d$, let us call a vertex of $G$ of degree $d$, at most $d$, and at least $d$ a \emph{$d$-vertex}, a \emph{$d^-$-vertex}, and a \emph{$d^+$-vertex} respectively.
For all vertex $v$, a \emph{$d$-neighbor}, a \emph{$d^-$-neighbor}, and a \emph{$d^+$-neighbor} of $v$ is a neighbor of $v$ that is a $d$-vertex, a $d^-$-vertex,
and a $d^+$-vertex respectively. For all $d$, let us call a face of $G$ of degree $d$, at most $d$, and at least $d$ a \emph{$d$-face}, a \emph{$d^-$-face},
and a \emph{$d^+$-face} respectively. For all set $S$ of vertices, an \emph{$S$-vertex} is a vertex that belongs to $S$, and an \emph{$S$-neighbor} of a vertex $v$
is a neighbor of $v$ that belongs to $S$.
For all set $S$ of vertices, let $G[S]$ denote the set of vertices induced by $S$, and $G-S=G[V(G)\setminus S]$. For convenience, we will note $G-v$ for $G-\acc{v}$.

Let us first prove some results on the structure of $G$, and then we will prove that $G$ cannot exist, thus proving the theorem.

\begin{lemma} \label{l_co}
$G$ is connected.
\end{lemma}

\begin{proof}
If $G$ is not connected, then every connected component of $G$ is smaller than $G$ and thus admits a $(0,6)$-coloring.
The union of these $(0,6)$-colorings gives a $(0,6)$-coloring of $G$, a contradiction.
\end{proof}

\begin{lemma} \label{l_1v}
$G$ has no $1$-vertex.
\end{lemma}

\begin{proof}
Let $v$ be a $1$-vertex and $w$ be the neighbor of $v$.
The graph $G-v$ admits a $(0,6)$-coloring since $G-v\prec G$.
We get a $(0,6)$-coloring $G$ by assigning to $v$ the color distinct from the color of $w$, a contradiction.
\end{proof}

\begin{lemma} \label{l_sb}
Every $7^-$-vertex of $G$ has a $8^+$-neighbor.
\end{lemma}

\begin{proof}
Let $v$ be a $7^-$-vertex with no $8^+$-neighbors. The graph $G-v$ admits a $(0,6)$-coloring since $G-v\prec G$.
If there is a neighbor $w$ of $v$ with no neighbor colored $0$, then we color $w$ with $0$.
Thus, we can assume that every neighbor of $v$ that is colored $k$ has a neighbor colored $0$ in $G-v$, and thus at most $5$ neighbors colored $k$ in $G-v$.
Also, we can assume that $v$ has at least one neighbor colored $0$, since otherwise $v$ can be colored $0$.
Thus, $v$ has at most $6$ neighbors colored $k$ and $v$ can be colored $k$, a contradiction.
\end{proof}

\begin{lemma} \label{l_3+bb}
Every vertex with degree at least $3$ and at most $7$ has two $8^+$-neighbors.
\end{lemma}

\begin{proof}
Suppose for contradiction that $G$ contains a $d$-vertex $v$ such that $3\le d\le7$ and such that $v$ has at most one $8^+$-neighbor.
By Lemma~\ref{l_sb}, $v$ has exactly one $8^+$-neighbor $w$. Let $w_1,\ldots,w_{d-1}$ be the other neighbors of $v$.
Let $H$ be the graph obtained from $G-v$ by adding $d-1$ $2$-vertices $v_1,\ldots,v_{d-1}$, such that for every $i\in\acc{1,d-1}$, $v_i$
is adjacent to $w$ and $w_i$. 

Notice that $H\prec G$ since $n_3(H)=n_3(G)-1$.
Moreover, every cycle of length $\ell$ in $H$ is associated a cycle of length $\ell$ or $\ell-2$ in $G$.
Therefore $H\in {\cal C}$, so $H$ has a $(0,6)$-coloring. 

If $w$ is colored $0$, then every $v_i$ is colored $k$,
coloring $v$ with $k$ leads to a $(0,6)$-coloring of $G$, a contradiction.
Therefore $w$ is colored $k$.

While at least one of the $w_i$'s has no neighbor colored $0$ in $G-v$, we color it $0$,
and color the corresponding $v_i$ with $k$ if it was colored $0$. 
By doing this, we keep a $(0,6)$-coloring of $H$.
We can thus assume that in $G-v$, every $w_i$ that is colored $k$ has a neighbor colored $0$ and thus at most five neighbors colored $k$.
If at least one of the $v_i$'s is colored $k$, then $w$ has at most five neighbors colored $k$ in $G-v$,
and assigning $k$ to $v$ gives a $(0,6)$-coloring of $G$.
Otherwise, every $v_i$ is colored $0$, every $w_i$ is colored $k$, and $w$ is colored $k$.
Thus we assign $0$ to $v$ to obtain a $(0,6)$-coloring of $G$, a contradiction.
\end{proof}

\begin{lemma} \label{l_23}
No $3$-vertex is adjacent to a $2$-vertex.
\end{lemma}

\begin{proof}
Let $w$ be a $3$-vertex adjacent to a $2$-vertex $v$, let $x_1$ and $x_2$ be the other two neighbors of $w$, and let $u$ be the other neighbor of $v$.
Let $H$ be the graph obtained from $G-\acc{v,w}$ by adding five $2$-vertices $v_1$, $v_2$, $w_1$, $w_2$, and $x$
which form the $8$-cycle $uv_1w_1x_1xx_2w_2v_2$. It is easy to check that $H$ is in $\cal C$.
By Lemmas~\ref{l_sb} and~\ref{l_3+bb}, $u$, $x_1$, and $x_2$ are $8^+$-vertices in $G$ and thus are $9^+$-vertices in $H$.
Since $w$ is in $G$ but not in $H$, $n_3(H)=n_3(G)-1$, so $H\prec G$.
Therefore $H$ has a $(0,6)$-coloring.

Suppose that $v_1$ and $v_2$ are both colored $0$. Then $w_1$, $w_2$, and $u$ are colored $k$.
We color $v$ with $0$ and $w$ with $k$. The number of neighbors of $x_1$ (resp. $x_2$) colored $k$ in $G$ is at most
the number of neighbors of $x_1$ (resp. $x_2$) colored $k$ in $H$. Thus we have a $(0,6)$-coloring of $G$, a contradiction.
Now we assume without loss of generality that $v_1$ is colored $k$.
We color $w$ with the color of $x$ and we color $v$ with $k$. The number of neighbors of $u$ (resp. $x_1$, $x_2$) colored $k$ in $G$ is at most
the number of neighbors of $u$ (resp. $x_1$, $x_2$) colored $k$ in $H$. Thus we have a $(0,6)$-coloring of $G$, a contradiction.
\end{proof}

A \emph{special face} is a $5$-face with three $2$-vertices and two non-adjacent $8^+$-vertices. See figure~\ref{fig1}, left.
A \emph{special configuration} is three $5$-faces sharing a common $3$-vertex adjacent to three $8^+$-vertices, such that all the other vertices of these faces are $2$-vertices.
See figure~\ref{fig1}, right. We say \emph{special structure} to speak indifferently about a special face or a special configuration.

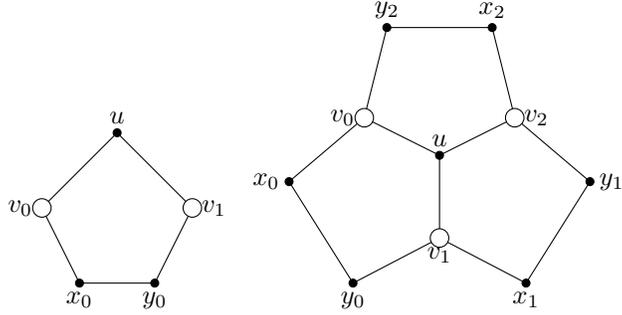
\begin{figure}[t]
\begin{center}
\begin{tikzpicture}
\coordinate (u) at (0,0) ;
\coordinate (v0) at (-1,-1) ;
\coordinate (v1) at (1,-1) ;
\coordinate (x0) at (-0.5,-2) ;
\coordinate (y0) at (0.5,-2) ;

\draw (u) -- (v0);
\draw (u) -- (v1);
\draw (v0) -- (x0);
\draw (v1) -- (y0);
\draw (x0) -- (y0);

\draw [fill=black] (u) circle (1.5pt) ;
\draw [fill=black] (x0) circle (1.5pt) ;
\draw [fill=black] (y0) circle (1.5pt) ;
\draw [fill=white] (v1) circle (3.5pt) ;
\draw [fill=white] (v0) circle (3.5pt) ;

\draw (u) node [above] {$u$} ;
\draw (v0) node [left] {$v_0$} ;
\draw (v1) node [right] {$v_1$} ;
\draw (x0) node [below] {$x_0$} ;
\draw (y0) node [below] {$y_0$} ;
\end{tikzpicture}
\begin{tikzpicture}
\coordinate (u) at (0,0) ;
\coordinate (v0) at (-1,0.5) ;
\coordinate (v1) at (0,-1.10) ;
\coordinate (v2) at (1,0.5) ;
\coordinate (x0) at (-2,-0.35) ;
\coordinate (y0) at (-1.15,-1.7) ;
\coordinate (x1) at (1.15,-1.7) ;
\coordinate (y1) at (2,-0.35) ;
\coordinate (x2) at (0.7,1.7) ;
\coordinate (y2) at (-0.7,1.7) ;

\draw (u) -- (v0);
\draw (u) -- (v1);
\draw (v0) -- (x0);
\draw (v1) -- (y0);
\draw (x0) -- (y0);
\draw (u) -- (v2);
\draw (v1) -- (x1);
\draw (v2) -- (y1);
\draw (x1) -- (y1);
\draw (v2) -- (x2);
\draw (v0) -- (y2);
\draw (x2) -- (y2);

\draw [fill=black] (u) circle (1.5pt) ;
\draw [fill=black] (x0) circle (1.5pt) ;
\draw [fill=black] (y0) circle (1.5pt) ;
\draw [fill=white] (v1) circle (3.5pt) ;
\draw [fill=black] (x1) circle (1.5pt) ;
\draw [fill=black] (y1) circle (1.5pt) ;
\draw [fill=black] (x2) circle (1.5pt) ;
\draw [fill=black] (y2) circle (1.5pt) ;
\draw [fill=white] (v2) circle (3.5pt) ;
\draw [fill=white] (v0) circle (3.5pt) ;

\draw (u) node [above] {$u$} ;
\draw (v0) node [left] {$v_0$} ;
\draw (v1) node [below] {$v_1$} ;
\draw (x0) node [left] {$x_0$} ;
\draw (y0) node [below] {$y_0$} ;
\draw (v2) node [right] {$v_2$} ;
\draw (x1) node [below] {$x_1$} ;
\draw (y1) node [right] {$y_1$} ;
\draw (x2) node [above] {$x_2$} ;
\draw (y2) node [above] {$y_2$} ;
\end{tikzpicture}
\caption{A special face (left) and a special configuration (right). \label{fig1}}
\end{center}
\end{figure}

Let us define a hypergraph $\widehat G$ whose vertices are the $8^+$-vertices of $G$ and the hyperedges correspond to the sets of $8^+$-vertices
contained in the same special structure.
For every vertex $v$ of $\widehat G$, let $\hat d(v)$ denote the degree of $v$ in $\widehat G$, that is the number of hyperedges containing $v$. 

\begin{lemma}\label{l_recol}
Let $\alpha$ be a special stucture, with the notations of Figure~\ref{fig1}.
Consider a $(0,6)$-coloring of $\alpha$.

We can change the color of the $x_i$'s, $y_i$'s and $u$ such that the $v_i$'s
have no more neighbors colored $k$ than before, and for all $i$, if $v_i$ is colored $k$, then $v_i$ has a neighbor colored $0$.
\end{lemma}

\begin{proof}
If all of the $v_i$'s are colored $0$, then there is noting to do. If they are all colored $k$, then we assign $0$ to $u$.
If one of the $v_i$'s, say $v_0$, is colored $0$ and another one, say $v_1$, is colored $k$, then $u$ and $x_0$ are colored $k$ and we assign $0$ to $y_0$.
Moreover, if $\alpha$ is a special configuration and $v_2$ is colored $k$, then $y_2$ is colored $k$ and we assign $0$ to $x_2$.
\end{proof}

\begin{lemma}\label{l_dhat}
For every vertex $v$ in $\widehat G$, $d(v)-\hat d(v)\ge 7$.
\end{lemma}

\begin{proof}
Let $v$ be a vertex that does not verify the lemma, i.e. such that $d(v)-\hat d(v)\le 6$. As $v$ is an $8^+$-vertex, $\hat d(v)\ge 1$.
Let $\alpha$ be a special structure incident to $v$ in $\widehat G$.
We use the notations of Figure~\ref{fig1}, with say $v=v_0$. The graph $G-x_0$ is smaller than $G$, thus it admits a $(0,6)$-coloring.
Since $G$ does not admit a $(0,6)$-coloring, $v_0$ is colored $k$ and $y_0$ is colored $0$. 
By Lemma~\ref{l_recol}, we can assume that $v$ has a neighbor colored $0$ in each of its special structures distinct from $\alpha$.
If $v_1$ is colored $0$, then $y_0$ is colored $k$, a contradiction. Thus $v_1$ is colored $k$.
If $\alpha$ is a special face, or if $v_2$ is colored $k$, then we assign $0$ to $u$.
If $\alpha$ is a special configuration and $v_2$ is colored $0$, then $x_2$ is colored $k$ and we assign $0$ to $y_2$.
In both cases, $v$ has at least $\hat d(v)$ neighbors colored $0$.
Thus $v$ has at most $d(v)-\hat d(v)\le 6$ neighbors colored $k$ and we can assign $k$ to $v$, a contradiction.
\end{proof}

\begin{lemma}\label{l_dhat2}
Every component of $\widehat G$ has at least one vertex $v$ such that $d(v)-\hat d(v)\ge 8$.
\end{lemma}

\begin{proof}
Suppose the lemma is false, and let $C$ be a component of $\widehat G$ that does not verify the lemma.
If $C$ has only one vertex, then this vertex is an $8^+$-vertex, which verifies $d(v)-\hat d(v)\ge 8$.
Therefore $C$ has at least one hyperedge, which corresponds to a special structure $\alpha$ of $G$.
By Lemma~\ref{l_dhat}, every vertex of $C$ verifies $d(v)-\hat d(v)=7$. We use the notations of Figure~\ref{fig1}.
The graph $G-\acc{x_0,y_0}$ is smaller than $G$, thus it admits a $(0,6)$-coloring.
Since $G$ admits no $(0,6)$-coloring, $v_0$ and $v_1$ are colored $k$.
If $\alpha$ is a special configuration and $v_2$ is colored $0$, then $x_2$ and $y_1$ are colored $k$ and we can color $y_2$ and $x_1$ with $0$.
Otherwise, we can color $u$ with $0$. Note that $v_0$ and $v_1$, as well as $v_2$ if it exists and is colored $k$, all have six neighbors colored $k$,
and by Lemma~\ref{l_recol}, we can assume that they all have at least one neighbor colored $0$ in each of their special structures besides $\alpha$. 

If one of the $v_i$'s, say $v_0$, has an additional neighbor colored $0$, it verifies $d(v_0)-\hat d(v_0)\ge 8$, a contradiction.
Thus, for every $v_i$, either $v_i$ is colored $0$ or $v_i$ has no neighbor colored $0$ outside of its special structures
and at most one neighbor colored $0$ in each special structure besides $\alpha$.

We uncolor $u$ and all the $x_i$'s and $y_i$'s, and let $H$ equal to $G$ where $u$, the $x_i$'s and the $y_i$'s are removed.
By symmetry, we only consider the vertex $v_0$.
The following procedure either assigns $0$ to $v_0$ or ensures that $v_0$ has two neighbors colored $0$ in one of its special structures:

\begin{itemize}
\item For each special structure $\beta$ containing $v_0$ and completely contained in $H$, we use the notations of Figure~\ref{fig1}, keeping the same vertex for $v_0$,
but changing the other ones for the vertices in $\beta$, and do the following:
\begin{itemize}
 \item By Lemma~\ref{l_recol}, we can assume that every $v_i$ colored $k$ has a neighbor colored $0$ in each of its special structures that are completely contained in $H$. 
 \item Suppose that one of the $8^+$-vertices of $\beta$ distinct from $v_0$, say $v_1$, has two neighbors colored $0$ in a special structure distinct from $\beta$ or a neighbor colored $0$ outside of its special structures.
 Since $d(v_1)-\hat d(v_1)=7$, $v_1$ has at most five neighbors colored $k$ outside of $\beta$ if $\beta$ is a special face, and at most four neighbors colored $k$ outside of $\beta$ if $\beta$ is a special configuration.
 We assign $k$ to $y_0$ $k$ and $0$ to $x_0$. If $v_2$ exists and is colored $0$, then we assign $0$ to $y_2$, and otherwise we assign $0$ to $u$. We end the procedure.
 \item We uncolor the $7^-$-vertices of $\beta$ and remove them from $H$. 
 \item For every $8^+$-vertex $w\ne v_0$ in $\beta$ colored $k$, we apply the procedure with $w$ instead of $v_0$.
 Now $w$ is colored $0$ or has two neighbors colored $0$ in the same special structure.
 \item We add back to $H$ the $7^-$-vertices of $\beta$. If $v_0$ is colored $0$, then we give them color $k$ if they are adjacent to a vertex colored $0$
 and we assign them $0$ otherwise,
 and we end the procedure. If $\beta$ is a special face and $v_1$ is colored $k$, or if $\beta$ is a special configuration and $v_1$ and $v_2$ are colored $k$,
 then we color $u$ and $x_0$ with $0$, we color the other $2$-vertices with $k$, and we end the procedure.
 Suppose $\beta$ is a special configuration, either $v_1$ or $v_2$, say $v_1$, is colored $k$, and the other one is colored $0$.
 We assign $0$ to $x_0$, $x_1$, and $y_2$, and $k$ to $u$, $y_0$, $y_1$, and $x_1$, and we end the procedure.
 Now all of the $v_i$'s distinct from $v_0$ are colored $0$. We color $x_0$ and $y_2$ (if it exists) with $0$ and we color the other $7^-$ vertices in $\beta$ with color $k$.
\end{itemize}
\item Now in each special structure containing $v_0$ and completely contained in $H$, all of the $8^+$-vertices distinct from $v_0$ is colored $0$.
 We assign $0$ to $v_0$ and $k$ to all of the neighbors of $v_0$.
\end{itemize}

Let us prove that the previous procedure terminates. It always calls itself iteratively on a graph with fewer vertices,
thus the number of nested iterations is bounded by the order of the initial graph.
Furthermore, each iteration of the procedure only does a bounded number of calls to the procedure (at most two).
That proves that the procedure terminates.

In the end, if one of the $v_i$'s is colored $k$, then it has at most five neighbors colored $k$ outside of $\alpha$ if $\alpha$ is a special face,
and at most four neighbors colored $k$ outside of $\alpha$ if $\alpha$ is a special structure.
If every $v_i$ is colored $k$, then color $u$ with color $0$ and the other $7^-$-vertex of $\alpha$ with color $k$. Otherwise, assign $k$ to $u$, and do the following:
\begin{itemize}
\item If every $v_i$ is colored $0$, then assign $k$ to the $x_i$'s and the $y_i$'s.

\item If $\alpha$ is a special face and one of the $v_i$'s, say $v_0$, is colored $0$ while the other one is colored $k$, then assign $k$ to $x_0$ and $0$ to $y_0$.

\item If $\alpha$ is a special structure, then assign $k$ to the $y_i$'s, and for all $i\in\acc{0,1,2}$,
 if $v_i$ is colored $k$, then assign $0$ to $x_i$, and if $v_i$ is colored $0$ then assign $k$ to $x_i$.
\end{itemize}
In all cases, we get a $(0,6)$-coloring of $G$, a contradiction.
\end{proof}

For each component $C$ of $\widehat G$, we choose a vertex $v$ in $C$ such that $d(v)-\hat d(v)\ge 8$ as the \emph{root} of $C$.
We then choose an orientation of the edges of $\widehat G$ such that the only vertices with no incoming edges are the roots
(for example do a breadth first search from the root of each component). For each $8^+$-vertex $v$, $v$ is said to \emph{sponsor} all of the special faces
that correspond to its outgoing edges in $\widehat G$.

We are now going to give some weight on the vertices and faces of the graph. Initially, for all $d$, every $d$-vertex has weight $d-4$,
and every $d$-face has weight $d-4$. Thus every face and every $4^+$-vertex has non-negative initial weight. 

We apply the following discharging procedure.
\begin{enumerate}
\item Every $8^+$-vertex gives weight $\frac{1}{2}$ to each of its $7^-$-neighbors, to each special face it sponsors,
and to the $3$-vertex of each special configuration it sponsors. Additionally, for every edge $vw$ where $v$ and $w$ are $8^+$-vertices, $v$ and $w$
each give $\frac{1}{4}$ to each of the faces containing the edge $vw$, and $\frac{1}{4}$ more to the face containing $vw$ if there is only one face containing $vw$. \label{d1}
\item For each $3^+$-vertex $v$ with degree at most $7$ in $G$, $v$ gives $\frac{1}{2}$ to each of its $2$-neighbors. Moreover, it gives $\frac{1}{2}$ to each of its $5$-faces
where it is adjacent to two $8^+$-vertices and where there are two $2$-vertices. \label{d2}
\item Each face $f$ gives $\frac{1}{4}$ to its $3^+$-vertices with degree at most $7$ that are consecutive to an $8^+$-vertex,
for each time they appear consecutively to an $8^+$-vertex in the boundary of $f$. \label{d3}
\item Each $5$-face gives $\frac{1}{4}$ to each of its $2$-vertices with no $2$-neighbor and $\frac{5}{8}$ to its $2$-vertices with a $2$-neighbor. \label{d4}
\item Each $7^+$-face gives $\frac{3}{4}$ to each of its $2$-vertices that belong to a $5$-face and have no $2$-neighbors, $\frac{7}{8}$ to each of its $2$-vertices
that belong to a $5$-face and have a $2$-neighbor, $\frac{1}{2}$ to each of its $2$-vertices that do not belong to a $5$-face and have no $2$-neighbors,
and $\frac{3}{4}$ to each of its $2$-vertices that do not belong to a $5$-face and have a $2$-neighbor. \label{d5}
\end{enumerate}

Let $\omega$ be the initial weight distribution, and let $\omega'$ be the final weight distribution, after the discharging procedure.

\begin{lemma}\label{l_vd}
Every vertex $v$ verifies $\omega'(v)\ge 0$.
\end{lemma}

\begin{proof}
Let $v$ be a vertex of degree $d$. We have $\omega(v)=d-4$.

\begin{itemize}
\item Suppose first that $d\ge 8$. The vertex $v$ gives $\frac{1}{2}$ to each of its $7^-$-neighbors and two times $\frac{1}{4}$
for each of its $8^+$-neighbors in Step~\ref{d1}, for a total of $\frac{d}{2}$. As $d\ge 8$, we have $\omega(v)=d-4\ge\frac{d}{2}$,
therefore if $v$ sponsors no special structure, then $\omega'(v)=d-4-\frac{d}{2}\ge 0$. 

Suppose $v$ sponsors a special structure. If $v$ sponsors all of its special structures, then $v$ is the root of its component in $\widehat G$,
thus $d-\hat d(v)\ge 8$, and thus $\omega'(v)=d-4-\frac{\hat d(v)}{2}-\frac{d}{2}=d-\hat d(v)-4-\frac{d-\hat d(v)}{2}\ge 0$.
If $v$ does not sponsor all of its special structures, then $d-\hat d(v)\ge 7$, and $\omega'(v)=d-4-\frac{\hat d(v)-1}{2}-\frac{d}{2}=
d-\hat d(v)-\frac{7}{2}-\frac{d-\hat d(v)}{2}\ge 0$.
\item Suppose now that $4\le d\le 8$. By Lemma~\ref{l_3+bb}, $v$ has at least two $8^+$-neighbors. The vertex $v$ only gives weight in Step~\ref{d2}.
Moreover, it gives at most $\frac{1}{2}$ to each of its $2$-neighbors plus $\frac{1}{2}$ for each pair of consecutive $8^+$-vertices in Step~\ref{d2}.
If $v$ has only $8^+$-neighbors, then it receives $\frac{d}{2}$ in Step~\ref{d1}, and gives at most $\frac{d}{2}$ in Step~\ref{d2}, so $\omega'(v)\ge \omega(v)=d-4\ge 0$.
Suppose $v$ has at least one $7^-$-neighbor. Let $d'\ge 2$ be the number of $8^+$-neighbors of $v$. The vertex $v$ receives $\frac{d'}{2}$ in Step~\ref{d1}.
It gives at most $\frac{d-d'}{2}$ to the $2$-vertices and at most $\frac{d'-1}{2}$ to the faces for a total of at most $\frac{d-d'}{2}+\frac{d'-1}{2}=\frac{d}{2}-\frac{1}{2}$
in step~\ref{d2}. It receives at least $\frac{d'}{4}$ in Step~\ref{d3}. We have $\omega'(v)\ge d-4-\frac{d}{2}+3\frac{d'}{4}+\frac{1}{2}\ge 0$, since $d'\ge 2$ and $d\ge 4$.
\item Suppose that $d=3$. By Lemma~\ref{l_3+bb}, $v$ has at least two $8^+$-neighbors, and by Lemma~\ref{l_23}, $v$ has no $2$-neighbors.
If $v$ has exactly two $8^+$-neighbors, then it receives $1$ in Step~\ref{d1}, gives $\frac{1}{2}$ in Step~\ref{d2}, and receives $\frac{3}{4}$
in Step~\ref{d3}, therefore $\omega'(v)\ge\frac{1}{4}>0$. If $v$ has three $8^+$-neighbors, then $v$ receives $\frac{3}{2}$ in Step~\ref{d1}
and an additional $\frac{3}{4}$ in Step~\ref{d3}, and it gives at most $1$ in Step~\ref{d2} unless it is in a special configuration,
in which case it gives at most $\frac{3}{2}$ in Step~\ref{d2} and receives $2$ in Step~\ref{d1}. Therefore if $v$ has three $8^+$-neighbors, then $\omega'(v)\ge\frac{1}{4}>0$.
\item Suppose that $d=2$. Note that $v$ cannot be in two $5$-faces since $G\in\cal C$. 
\begin{itemize}
\item If $v$ is in a $5$-face and adjacent to another $2$-vertex, then it receives $\frac{1}{2}$ from its $8^+$-neighbor in Step~\ref{d1},
$\frac{5}{8}$ from its $5$-face in Step~\ref{d4}, and $\frac{7}{8}$ from its other face in Step~\ref{d5}. 
\item If $v$ is in a $5$-face and adjacent to no other $2$-vertex, then it receives $1$ from its $3^+$-neighbors in Steps~\ref{d1} and~\ref{d2},
$\frac{1}{4}$ from its $5$-face in Step~\ref{d4}, and $\frac{3}{4}$ from its other face in Step~\ref{d5}. 
\item If $v$ is not in a $5$-face and is adjacent to another $2$-vertex, then it receives $\frac{1}{2}$ from its $8^+$-neighbor in Step~\ref{d1},
and $\frac{3}{2}$ from its faces in Step~\ref{d5}. 
\item If $v$ is in a $5$-face and adjacent to no other $2$-vertex, then it receives $1$ from its $3^+$-neighbors in Steps~\ref{d1} and~\ref{d2},
and $1$ from its faces in Step~\ref{d5}. 
\end{itemize}
In all cases, $v$ receives $2$ over the procedure, and thus $\omega'(v)=2-4+2=0$.
\end{itemize}
\end{proof}

\begin{lemma}\label{l_fd}
Every face $\alpha$ satisties $\omega'(\alpha)\ge 0$.
\end{lemma}

\begin{proof}
Let $\alpha$ be a vertex of degree $d$. We have $\omega(\alpha)=d-4$.

\begin{itemize}
\item Suppose $d=5$. If $\alpha$ is a special face, then it receives $\frac{1}{2}$ in Step~\ref{d1} and gives $\frac{1}{4}+2\cdot\frac{5}{8}=\frac{3}{2}$ in Step~\ref{d4}. 

If $\alpha$ has no two consecutive $2$-vertices, then it gives at most $\frac{1}{4}$ to its small vertices over Steps~\ref{d3} and~\ref{d4},
and does not actually give anything unless one of its vertices is an $8^+$-vertex, and thus gives at most $1$ overall. 

If $\alpha$ has two consecutive $2$-vertices and its three other vertices are $8^+$-vertices, then it receives $1$ in Step~\ref{d1}
and gives at most $2\cdot\frac{5}{8}=\frac{5}{4}\le 2$ overall.

The only remaining case is when $\alpha$ has, in this consecutive order, two $2$-vertices, an $8^+$-vertex,
a $3^+$-vertex with degree at most $7$, and another $8^+$-vertex. In this case, $\alpha$ receives $\frac{1}{2}$ in Step~\ref{d2},
and gives $2\cdot\frac{5}{8}+\frac{1}{4}=\frac{3}{2}$ over Steps~\ref{d3} and~\ref{d4}.

In all cases, $\omega'(\alpha)\ge 1-1=0$.
\item Suppose $d=7$. Note that if there are two adjacent $2$-vertices in $\alpha$, then these two vertices are not in a $5$-face,
otherwise there would be a cycle of length $6$ in $G$. The face $\alpha$ has an initial charge of $3$, gives at most $\frac{3}{4}$
to its $7^-$-vertices that are adjacent to an $8^+$-vertex in $\alpha$ and nothing to its other vertices. There can be at most four of these vertices.
Therefore $\omega'(\alpha)\ge 3-4\cdot\frac{3}{4}=0$.
\item Suppose $d=8$. Note that at most one pair of adjacent $2$-vertices is in a $5$-face, otherwise there would be a cycle of length $6$ in $G$.
The face $\alpha$ has an initial charge of $4$, gives at most $\frac{7}{8}$ to its $7^-$-vertices that are adjacent to an $8^+$-vertex in $\alpha$,
and nothing to its other vertices. There can be at most five of these vertices, and at most two are given $\frac{7}{8}$,
the other being given at most $\frac{3}{4}$. Therefore $\omega'(\alpha)\ge 4-2\cdot\frac{7}{8}-3\cdot\frac{3}{4}=0$.
\item Suppose $d=9$. Note that at most two pairs of adjacent $2$-vertices are in a $5$-face, otherwise there would be a cycle of length $6$ in $G$.
The face $\alpha$ has an initial charge of $5$, gives at most $\frac{7}{8}$ to its $7^-$-vertices that are adjacent to an $8^+$-vertex in $\alpha$,
and nothing to its other vertices. There can be at most six of these vertices, at most four are given $\frac{7}{8}$, and the others are given at most $\frac{3}{4}$.
Therefore $\omega'(\alpha)\ge 5-4\cdot\frac{7}{8}-2\cdot\frac{3}{4}=0$.
\item Suppose $d\ge 10$. 
The face $\alpha$ has an initial charge of $d-4$, gives at most $\frac{7}{8}$ to its $7^-$-vertices that are adjacent to an $8^+$-vertex in $\alpha$,
and nothing to its other vertices. There can be at most $d-4$ of these vertices, therefore $\omega'(\alpha)\ge d-4-(d-4)\cdot\frac{7}{8}>0$.
\end{itemize}
\end{proof}

By Euler's formula, since $G$ is connected by Lemma~\ref{l_co} and has at least one vertex, $n+f-m=2$. The initial weight of the graph is 
$\sum_{v\in V(G)}\omega(v)+\sum_{\alpha\in F(G)}\omega(\alpha) = \sum_{v\in V(G)} (d(v)-4)+\sum_{\alpha \in F(G)} (d(\alpha)-4)=
\sum_{v\in V(G)} d(v)+\sum_{\alpha \in F(G)} d(\alpha)-4n-4f=4m-4n-4f=-8 < 0$. Therefore the initial weight of the graph is negative,
thus the final weight of the graph is negative. Since by Lemmas~\ref{l_vd} and~\ref{l_fd}, the final weight of every face and
every vertex is non-negative, we get a contradiction. This completes the proof of Theorem~\ref{t+}.

\section{Proof of Theorem~\ref{t-}} \label{sec:t-}

Let $k\ge3$ be a fixed integer. Suppose that there exists a graph in $\cal C$ that is not $(0,k)$-colorable.
We consider such a graph $H_k$ that is minimal according to $\preceq$.
By adapting the proofs of Lemmas~\ref{l_co},~\ref{l_1v}, and~\ref{l_sb}, we obtain that the minimum degree of $H_k$ is at least two
and every $(k+1)^-$-vertex in $H_k$ is adjacent to a $(k+2)^+$-vertex. Suppose for contradiction that $H_k$ contains no $2$-vertex.
We consider the discharging procedure such that the initial chage of every vertex is equal to its degree
and every $5^+$-vertex gives $\tfrac13$ to every adjacent $3$-vertex.
Then the final charge of a $3$-vertex is at least $3+\tfrac13=\tfrac{10}3$, the final charge of a $d$-vertex with $d\ge k+2$ is at least
$d-d\times\tfrac13\ge\tfrac 2d/3\ge\tfrac 2(k+2)/3\ge\tfrac{10}3$, and the final charge of every remaining vertex is at least $4>\tfrac{10}3$.
This implies that the maximum average degree of $H_k$ is at least $\tfrac{10}3$, which is a contradiction since $H_k$ is a planar graph with girth at least $5$.
Thus, $H_k$ contains a $2$-vertex $v$ adjacent to the vertices $u_1$ and $u_5$.

By minimality of $H_k$, $H_k-v$ is $(0,k)$-colorable, every $(0,k)$-coloring of $H_k-v$ is such that $u_1$ and $u_5$ get distinct colors, and
the vertex in $\acc{u_1,u_5}$ that is colored $k$ has exactly $k$ neighbors that are colored $k$.

Consider the graph $H'_k$ obtained from $H_k-v$ by adding three $2$-vertices $u_2$, $u_3$, and $u_4$ which form a path $u_1u_2u_3u_4u_5$.
Notice that $H'_k$ is $(0,k)$-colorable and that in every $(0,k)$-coloring of $H'_k$ is such that $u_3$ is colored $k$ and is adjacent to exactly
one vertex colored $k$. It is easy to see that $H'_k$ is in $\cal C$.

We are ready to prove that deciding whether a graph in $\cal C$ is $(0,k)$-colorable is NP-complete.
The reduction is from the NP-complete problem of deciding whether a planar graph with girth at least $9$ is $(0,1)$-colorable~\cite{EMOP:2013}.
Given in instance $G$ of this problem, we construct a graph $G'\in\cal C$, as follows
For every vertex $v$ in $G$, we add $k-1$ copies of $H'_k$ and we add an edge between $v$ and the vertex $u_3$ of each these copies.
Notice that $G'$ is in $\cal C$ since $G'$ is planar and every cyle of length at most $8$ is contained in a copy of $H'_k$ which is in $\cal C$.
Notice that a $(0,1)$-coloring of $G$ can be extended to a $(0,k)$-coloring of $G'$.
Conversely, a $(0,k)$-coloring of $G'$ induces a $(0,1)$-coloring of $G$.
So $G$ is $(0,1)$-colorable if and only if $G'$ is $(0,k)$-colorable.

\section{A graph in $\cal C$ that is not $(0,3)$-colorable} \label{sec:ce}

Consider the graph $F_{x,y}$ depicted in Figure~\ref{04}. Suppose for contradiction that $F_{x,y}$ admits a $(0,3)$-coloring such that
all the neighbors of $x$ and $y$ are colored $0$ (the white vertices in the picture). Then the neighbors of those white vertices are
colored $k$. We consider the $8$ big vertices. Each of them is colored $k$ and is ajacent to two vertices colored $k$.
For every pair of adjacent red vertices, at least one of them is colored $k$. Notice that every red vertex is adjacent to a big vertex.
Since there are $9$ pairs of adjacent red vertices, there exists a big vertex that is adjacent to at least two red vertices colored $k$.
This big vertex is thus adjacent to four vertices colored $k$, which is a contradiction.

In the graph depicted in Figure~\ref{bb}, every dashed line represent a copy of $F_{x,y}$ such that the extremities are $x$ and $y$.
Suppose for contradiction that this $(C_3,C_4,C_6)$-free planar graph admits a $(0,3)$-coloring.
Each of the two drawn edges has at least one extremity colored $k$.
Thus, there exist two vertices $u$ and $v$ colored $k$ that are linked by $7$ copies of $F_{x,y}$.
Since at most $3$ neighbors of $u$ and at most $3$ neighbors of $v$ can be colored $k$,
one these $7$ copies of $F_{x,y}$ is such that all the neighbors of $x$ and $y$ are colored $0$.
This is contradiction proves Theorem~\ref{t-}.

Following the proof above, we see that if we remove the green parts in Figures~\ref{04} and~\ref{bb},
we obtain a planar graph with girth $7$ that is not $(0,2)$-colorable.
A graph with such properties is already known~\cite{MO14}, but this new graph is smaller ($184$ vertices instead of $1304$)
and the proof of non-$(0,2)$-colorability is simpler.

\begin{figure}
\begin{minipage}[b]{0.5\linewidth}
\begin{center}
\includegraphics[width=60mm]{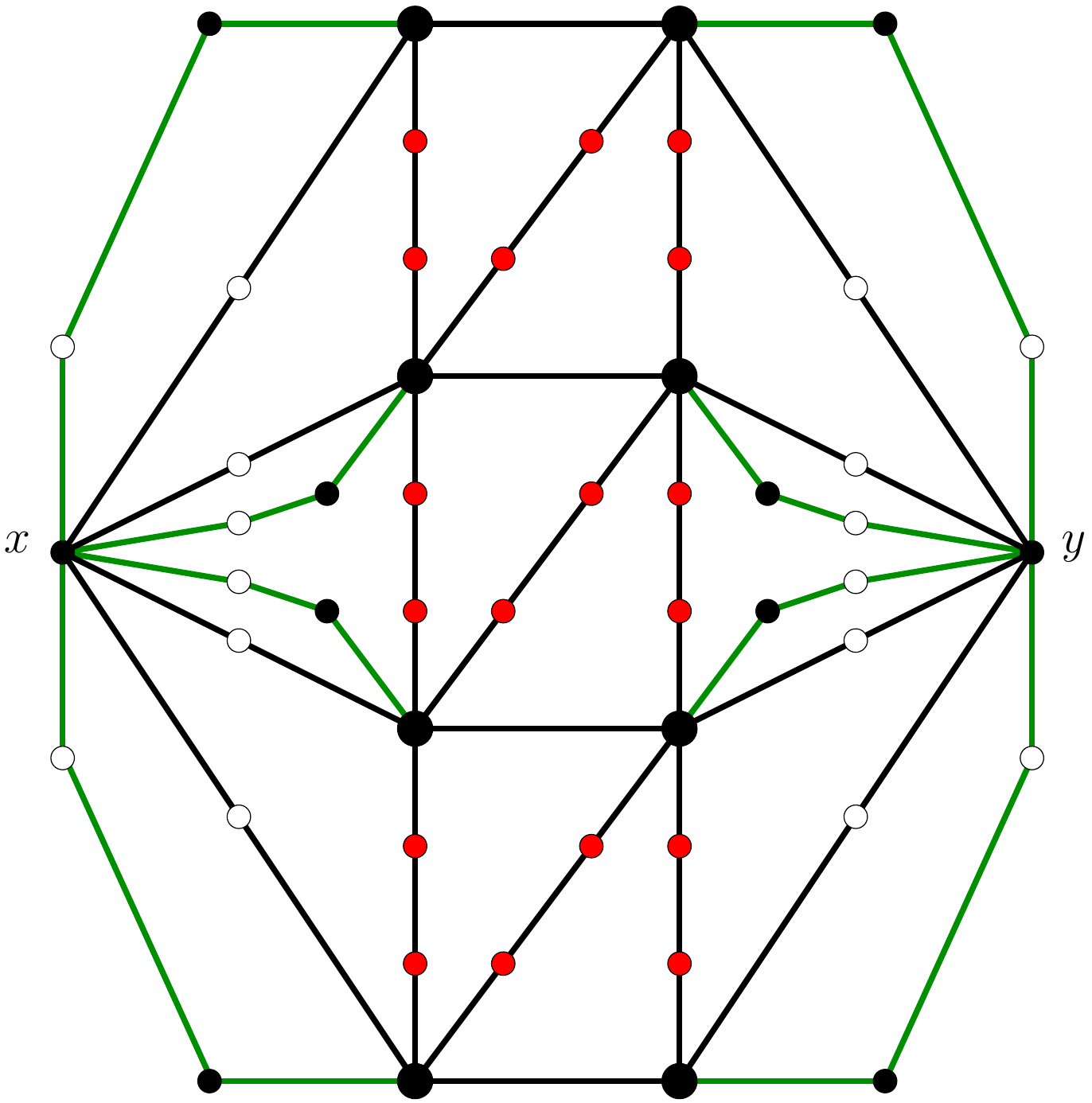}
\caption{The forcing gadget $F_{x,y}$.}
\label{04}
\end{center}
\end{minipage}
\begin{minipage}[b]{0.5\linewidth}
\begin{center}
\includegraphics[width=40mm]{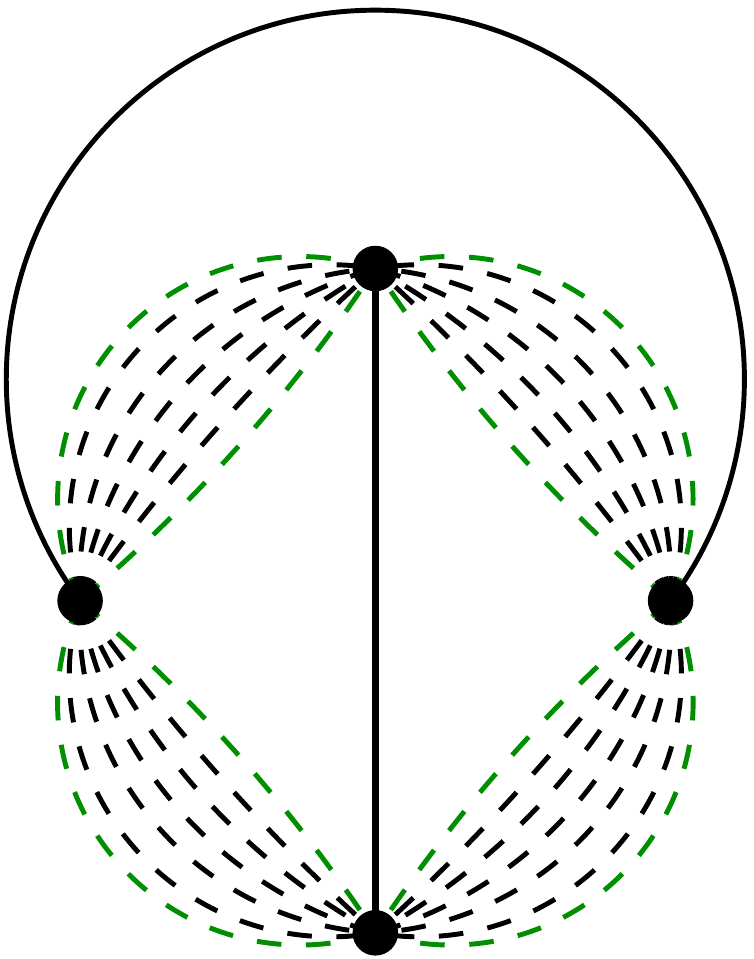}
\caption{The non-$(0,3)$-colorable graph in $\cal C$.}
\label{bb}
\end{center}
\end{minipage}
\end{figure}


\begin{thebibliography}{99}
\bibitem{CLO} I.~Choi, C-H.~Liu S.~Oum.
\newblock Characterization of cycle obstruction sets for improper coloring planar graphs.
\newblock \url{http://mathsci.kaist.ac.kr/~sangil/pdf/2016balanced.pdf}

\bibitem{EMOP:2013} L.~Esperet, M.~Montassier, P.~Ochem, and A.~Pinlou.
\newblock A complexity dichotomy for the coloring of sparse graphs.
\newblock \emph{J. Graph Theory} {\bf 73(1)} (2013), 85--102.

\bibitem{MO14} M.~Montassier and P.~Ochem.
Near-colorings: non-colorable graphs and NP-completeness.
\emph{Electron. J. Comb.} \textbf{22(1)} (2015), \#P1.57.

\end{thebibliography}
\end{document}